\documentclass[oneside,10pt]{article}          
\usepackage[b5paper]{geometry}	    

\usepackage{graphicx}
\usepackage{epstopdf}
\usepackage{amsfonts,amsmath,latexsym,amssymb} 
\usepackage{amsthm}                
\usepackage{mathrsfs,upref}         
\usepackage{mathptmx}		    
\usepackage{mia}	            

\usepackage[final]
{showkeys}

\usepackage{fancybox,framed}

\usepackage{hyperref}

\newtheorem{theorem}{Theorem}           
               
\newtheorem{corollary}[theorem]{Corollary}

\theoremstyle{definition}

\newtheorem{remark}[theorem]{Remark}

\numberwithin{equation}{section}       

\newcommand{\tr}{\operatorname{tr}}
\newcommand{\al}{\alpha}

\newcommand{\si}{\sigma}

\newcommand{\la}{\lambda}

\newcommand{\De}{\Delta}
\newcommand{\ep}{\varepsilon}

\renewcommand{\Psi}{\overline{\Phi}}

\newcommand{\R}{\mathbb{R}}

\newcommand{\J}{\mathcal J}

\renewcommand{\H}{\mathcal H}

\renewcommand{\L}{\mathcal{L}}

\newcommand{\Ts}{\mathrm{Ts}}
\newcommand{\RR}{\mathrm{R}}

\newcommand{\down}{{\downarrow}}
\newcommand{\up}{{\uparrow}}

\newcommand{\tp}{{\tilde{p}}}
\newcommand{\tq}{{\tilde{q}}}
\newcommand{\tf}{{\tilde{f}}}
\newcommand{\tep}{{\tilde{\ep}}}

\newcommand{\trho}{{\tilde{\rho}}}
\newcommand{\tsi}{{\tilde{\si}}}
\renewcommand{\le}{\leqslant}
\renewcommand{\ge}{\geqslant}

\begin{document}

\title[Modulus of continuity of the quantum $f$-entropy with respect to the trace distance]{Modulus of continuity of the quantum $f$-entropy with respect to the trace distance}


\author{Iosif Pinelis}

\address{Department of Mathematical Sciences\\
Michigan Technological University\\
Houghton, Michigan 49931, USA\\
\email{ipinelis@mtu.edu}}

%

\CorrespondingAuthor{Iosif Pinelis}


\date{\today}                               

\keywords{modulus of continuity; density matrices; trace distance; von Neumann entropy; quantum $f$-entropy; convex functions; Schur majorization}

\subjclass{Primary 15A42, 26B25, 26D15, 47A30, 47A63, 47B06, 94A17; secondary 47A60, 47B15, 47B65, 47A70}



\begin{abstract}
A well-known result due to Fannes is a certain upper bound on the modulus of continuity of the von Neumann entropy with respect to the trace distance between density matrices; this distance is the maximum probability of distinguishing between the corresponding quantum states. Much more recently, Audenaert obtained an exact expression of this modulus of continuity. 

In the present note, Audenaert's result is extended to a broad class of entropy functions indexed by arbitrary continuous convex functions $f$ in place of the Shannon--von Neumann function $x\mapsto x\log_2x$. The proof is based on the Schur majorization. 
\end{abstract}

\maketitle



\section{Summary and discussion}
Let $\rho$ and $\si$ be two density matrices of a finite dimension $d\ge2$, that is, two positive-semidefinite Hermitian linear operators of trace $1$ acting on a $d$-dimensional Hilbert space $\H$. The trace distance between $\rho$ and $\si$ is 
\begin{equation}\label{eq:T}
	T(\rho,\si):=\frac12\,\tr|\rho-\si|
	=\frac12\,\sum_{i=1}^d|\la_i(\rho-\si)|=\sum_{i=1}^d\max\big(0,\la_i(\rho-\si)\big).  
\end{equation}
Here and in what follows, as usual, 
$\tr$ denotes the trace, $\la_1(\tau)\ge\cdots\ge\la_d(\tau)$ are the eigenvalues of a Hermitian linear operator $\tau$ on $\H$, and, for any real-valued function $f$ defined on the spectrum of $\tau$, the action of $f$ on $\tau$ is defined by the formula 
\begin{equation*}
	f(\tau):=\sum_{i=1}^d f(\la_i(\tau))P_i(\tau), 
\end{equation*}
given the spectral decomposition $\tau=\sum_{i=1}^d \la_i(\tau)P_i(\tau)$ of $\tau$, with appropriate orthogonal projectors $P_i(\tau)$, so that $\sum_{i=1}^d P_i(\tau)$ is the identity operator on $\H$; in particular, $|\tau|=\sum_{i=1}^d |\la_i(\tau)|P_i(\tau)$.  

One may also write 
\begin{equation}\label{eq:T=sup}
	T(\rho,\si)=\sup_P\tr\big(P(\rho-\si)\big),  
\end{equation}
where the supremum is taken over all orthogonal projectors $P$ of the Hilbert space $\H$; cf.\ e.g.\ \cite[Ch.\ 9]{nielsen-chuang}. Thus, in view of Gleason's theorem \cite{gleason}, which states that any natural assignment of probabilities to measurement outcomes must follow the Born rule, one sees that the trace distance $T(\rho,\si)$ is the maximum probability of distinguishing between the two quantum states given by $\rho$ and $\si$. 

By the norm inequality or by (\ref{eq:T=sup}), $0\le T(\rho,\si)\le1$. 
Moreover, it is easy to see that $T(\rho,\si)$ can take any value in the interval $[0,1]$. Indeed, 
suppose e.g.\ that the operators $\rho$ and $\si$ are commuting, with the same eigenbasis and the corresponding eigenvalues $p_j$ and $q_j$ given by the formulas 
$p_1:=t$, $p_2=1-t$, $q_1:=1-t$, $q_2=t$, and $p_3=q_3=\dots=p_d=q_d=0$. Then, letting $t$ vary from $0$ to $1/2$, 
we see that $T(\rho,\si)=\frac12\,\sum_{k=1}^d|p_k-q_k|$ will continuously vary from $1$ to $0$. 

The von Neumann entropy of a density matrix $\rho$ is 
\begin{equation*}
	S(\rho):=-\tr(\rho\log_2\rho), 
\end{equation*}
with $0\log_20:=0$. 
Audenaert \cite{audenaert} showed that 
\begin{equation}\label{eq:aud}
	|S(\rho)-S(\si)|\le \De_d(\ep):=h(\ep)+\ep\log_2(d-1), 
\end{equation}
where 
\begin{equation}\label{eq:ep}
	\ep:=T(\rho,\si)  
\end{equation}
and 
\begin{equation*}
	h(\ep):=-\ep\log_2\ep-(1-\ep)\log_2(1-\ep). 
\end{equation*}
As pointed out in \cite{audenaert}, the upper bound $\De_d(\ep)$ on $|S(\rho)-S(\si)|$ is better (that is, smaller) than the well-known bound due to Fannes 
\cite{fannes73,nielsen-chuang}. Moreover, as was noted in \cite{audenaert}, the bound $\De_d(\ep)$ on $|S(\rho)-S(\si)|$ is the best possible one in terms of $\ep$ and $d$. 

In this note, 
inequality (\ref{eq:aud}) is extended to general continuous convex functions of density matrices instead of the convex function $p\mapsto p\log_2p$, as follows: 

\begin{theorem}\label{th:}
Take any continuous convex function $f\colon[0,1]\to\R$ and consider the corresponding generalized $f$-entropy 
\begin{equation}\label{eq:S}
	S_f(\rho):=-\tr f(\rho) 
\end{equation}
of a density matrix $\rho$. 
Then 
\begin{align}
|S_f(\rho)-S_f(\si)|&\le\De_{f;d}(\ep) \label{eq:}\\ 
&:=f(1)-f(1-\ep)-(d-1)\Big(f\Big(\frac\ep{d-1}\Big)-f(0)\Big).  \nonumber  
\end{align}
The bound $\De_{f;d}(\ep)$ on $|S_f(\rho)-S_f(\si)|$ is exact for each $\ep\in[0,1]$, as it is attained by $|S_f(\rho)-S_f(\si)|$ for some density matrices $\rho$ and $\si$ such that $T(\rho,\si)=\ep$.
\end{theorem} 

The necessary proofs will be given in Section~\ref{proofs}.

In the particular case when $f(p)=p\ln p$ (with $f(0)=0$), the exact bound $\De_{f;d}(\ep)$
coincides with the Audenauert bound $\De_d(\ep)$. 

Almost immediately from Theorem~\ref{th:}, we obtain the following. 

\begin{corollary}\label{cor:}
The modulus of continuity of the generalized $f$-entropy $S_f(\cdot)$ with respect to the trace distance
is given by the formula 
\begin{align*}
	\sup_{(\rho,\si)\colon\, T(\rho,\si)\le\ep}|S_f(\rho)-S_f(\si)|
&=\max_{(\rho,\si)\colon\, T(\rho,\si)\le\ep}|S_f(\rho)-S_f(\si)| \\ 
	&=\De_{f;d}\big(\min(\ep,1-\tfrac1d)\big) 
\end{align*}
for all $\ep\in[0,1]$. 
\end{corollary}

The von Neumann entropy is a special case of the generalized $f$-entropy, with $f(x)=x\log_2 x$. 
Another special case of the generalized $f$-entropy is the Tsallis entropy \cite{tsallis88,gell-mann--tsallis,tsallis10} 
\begin{equation*}
	\Ts_\al(\rho):=\frac{1-\tr(\rho^\al)}{\al-1}
\end{equation*}
for real $\al>1$, corresponding to the continuous convex function $f$ given by the formula $f(x):=(x^\al-x)/(\al-1)$ for real $x\in[0,1]$. 
The special case of the Tsallis entropy with $\al=2$ is also known as the Gini--Simpson index \cite{jost} or the Gibbs--Martin/Blau index \cite{Blau18} or the expected heterozygosity \cite{sbordoni}. 
The Tsallis entropy $\Ts_\al(\rho)$ is related with the Renyi entropy 
\begin{equation*}
	\RR_\al(\rho):=\frac{\log_2\tr(\rho^\al)}{1-\al}
\end{equation*}
by the strictly increasing one-to-one transformation given by the formula 
\begin{equation*}
	\Ts_\al(\rho)=\frac{1-2^{(1-\al)\RR_\al(\rho)}}{\al-1}. 
\end{equation*}

In the non-quantum case, that is, for a ``probability distribution'' $(p_1,\dots,p_d)$ with nonnegative $p_1,\dots,p_d$ such that $p_1+\dots+p_d=1$, the notion of the $f$-entropy \break 
$-\sum_{k=1}^d f(p_k)$, again for a convex function $f$, was considered in \cite{arimoto}. 

The special case of Theorem~\ref{th:} corresponding to the Tsallis entropy was obtained in \cite{zhang07}, by a different method, involving probabilistic coupling. 


In the following corollary, obtained from Theorem~\ref{th:} mainly by simple rescaling, the unit-trace condition on $\rho$ and $\si$ is dropped. The resulting statement may be of some convenience. It will actually be used in the proof of Theorem~\ref{th:}. 

\begin{corollary}\label{cor:t}
Let $\rho$ and $\si$ be two positive-semidefinite Hermitian linear operators of the same trace $t$ acting on a Hilbert space of a finite dimension $d\ge2$. Take any continuous convex function $f\colon[0,t]\to\R$. 
Then $\ep\in[0,t]$ and 
\begin{align}	
|S_f(\rho)-S_f(\si)|&\le\De_{t;f;d}(\ep) \label{eq:t} 
:=f(t)-f(t-\ep)-(d-1)\Big(f\Big(\frac\ep{d-1}\Big)-f(0)\Big),  
\end{align}
with $S_f$, $T$, and $\ep$ still as defined in \eqref{eq:S}, \eqref{eq:T}, and \eqref{eq:ep}.  
Also, $\De_{t;f;d}(\ep)$ is nondecreasing in $t$ and in $d$. 
\end{corollary}

\begin{remark}\label{rem:}
Just as 
the bound $\De_{f;d}(\ep)=\De_{1;f;d}(\ep)$ is exact in the setting of Theorem~\ref{th:}, the bound $\De_{t;f;d}(\ep)$ is exact in the slightly more general setting of  Corollary~\ref{cor:t}. 
However, somewhat surprisingly and in contrast with the last sentence in Corollary~\ref{cor:t}, $\De_{t;f;d}(\ep)$ is not monotonic in $\ep\in[0,t]$, for any real $t>0$, as can be easily seen from the proofs of Corollaries~\ref{cor:t} and \ref{cor:}. 
\end{remark}

\section{Proofs}\label{proofs}

First here, let us deduce Corollary~\ref{cor:t} from Theorem~\ref{th:} or, more specifically, from inequality \eqref{eq:}: 

\begin{proof}[Proof of Corollary~\ref{cor:t}]
That $\ep\in[0,t]$ follows 
by the norm inequality. If $t=0$ then $\rho=\si=0$ and hence $\ep=0$, so that inequality \eqref{eq:t} is trivial. Assume now that $t>0$. 
Let $\trho:=\rho/t$, $\tsi:=\si/t$, $\tep:=T(\trho,\tsi)=\ep/t$, and $\tf(u):=f(tu)$ for all real $u$. 
Then $\tr\trho=1=\tr\tsi$, 
$S_f(\rho)=S_\tf(\trho)$, $S_f(\si)=S_\tf(\tsi)$, and $\De_{\tf;d}(\tep)=\De_{t;f;d}(\ep)$. So, \eqref{eq:t} is obtained by using \eqref{eq:} with $\trho$, $\tsi$, $\tep$, and $\tf$ instead of $\rho$, $\si$, $\ep$, and $f$, respectively. 
The last sentence in Corollary~\ref{cor:t} follows by the convexity of $f$. 
\end{proof}

\begin{proof}[Proof of Theorem~\ref{th:}]
%


We shall prove inequality \eqref{eq:} by induction on $d$. 

By approximation, without loss of generality (wlog) the function $f$ is \emph{strictly} convex and differentiable.

The proof uses the powerful majorization tool; cf.\ e.g.\ \cite[Definition~A.1]{marsh-ol09}. 
%
For any natural $n$, we say that a vector $a=(a_1,\dots,a_n)\in\R^n$  majorizes (in the Schur sense) a vector $b=(b_1,\dots,b_n)\in\R^n$ and write $a\succcurlyeq b$ if for the corresponding decreasing rearrangements $a^\down=(a_{n:1},\dots,a_{n:n})$ and $b^\down=(b_{n:1},\dots,b_{n:n})$ of the vectors $a$ and $b$ we have $\sum_{i=1}^n a_{n:i}=\sum_{i=1}^n b_{n:i}$ and $\sum_{i=1}^j a_{n:i}\ge\sum_{i=1}^j b_{n:i}$ for all $j=1,\dots,n-1$. One may note here that 
\begin{equation}\label{eq:order}
	a_{n:j}=\max_{J\in\J_{n,j}
	}\min_{i\in J}a_i
\end{equation}
$j\in[n]:=\{1,\dots,n\}$, where $\J_{n,j}$ denotes the set of all subsets of cardinality $j$ of the set $[n]$; cf.\ e.g.\ \cite[formula (8.2)]{semi-mod-publ}. 

As in \cite{audenaert}, 
we now invoke the fundamental double inequality 
\begin{equation}\label{eq:ky fan}
	\tfrac12\,\|p^\down-q^\down\|_1=:\ep^{\down\down}
	\le \ep=T(\rho,\si)\le\ep^{\up\down}:=\tfrac12\,\|p^\up-q^\down\|_1,
\end{equation}
where $\|\cdot\|_1$ is the $\ell^1$ norm on $\R^d$, 
$p^\up$ is the vector of the eigenvalues $p_1,\dots,p_d$ of $\rho$ sorted in the ascending order, 
and 
$q^\down$ is the vector of the eigenvalues $q_1,\dots,q_d$ of $\si$ sorted in the descending order. 
The double inequality \eqref{eq:ky fan} is a special, trace-norm case of the corresponding more general result for unitarily invariant norms; see e.g.\ the double inequality \cite[(IV.62)]{bhatia}. 

Note that for $\ep\in(0,1)$ the 
derivative in $\ep$ of $\De_{f;d}(\ep)$ is $f'(1-\ep)-f'(\frac\ep{d-1})\ge0$ if $1-\ep\ge\frac\ep{d-1}$, that is, if $\ep\le1-\frac1d$. So, $\De_{f;d}(\ep)$ is nondecreasing in $\ep\in[0,1-\frac1d]$. Similarly, $\De_{f;d}(\ep)$ is nonincreasing in $\ep\in[1-\frac1d,1]$. 
Now it follows by \eqref{eq:ky fan} that 
\begin{equation}\label{eq:incr-decr}
	\De_{f;d}(\ep)\ge\min\big(\De_{f;d}(\ep^{\down\down}),\De_{f;d}(\ep^{\up\down})\big). 
\end{equation}
(The argument presented in this paragraph is missing in \cite{audenaert}.) 

\begin{framed}
\noindent At this point, let us ``forget'' the definition of $\ep$ in \eqref{eq:ep} and, instead,
take any $\ep\in(0,1)$. 
\end{framed} 
In view of \eqref{eq:incr-decr}, to prove inequality \eqref{eq:}, it is enough to show that 
\begin{equation}\label{eq:sup}
	D_f(p,q)\overset{\text{(?)}}\le\De_{f;d}(\ep)
\end{equation}
for all $(p,q)\in P_{d;\ep}$, 
where 
\begin{align}
\label{eq:de}
	D_f(p,q)&:=\sum_1^d f(p_i)-\sum_1^d f(q_i),  \\
	P_{d;\ep}&:=\{(p,q)\in P_d\colon E(p,q)=\ep\}, \notag \\
	P_d&:=S_d\times S_d, \notag \\
	S_d&:=\Big\{p=(p_1,\dots,p_d)\in[0,\infty)^d\colon\sum_1^d p_i=1\Big\}, \notag \\ 
		E(p,q)&:=\frac12\sum_1^d|p_i-q_i|.  \notag
\end{align}

In the case $d=2$, take any
$(p,q)\in P_{d;\ep}$, so that $p=(a,1-a)\in S_d$ and $q=(b,1-b)\in S_d$ for some $a$ and $b$ in $[0,1]$ such that $\ep=E(p,q)=|a-b|$. 
Wlog, $a\ge b$, and hence $\ep=E(p,q)=a-b\ge0$. Therefore and because $f$ is convex, 
\begin{align*}
	D_f(p,q)&=[f(a)-f(b)]-[f(1-b)-f(1-a)] \\ 
	&=[f(a)-f(a-\ep)]-[f(1-a+\ep)-f(1-a)] \\ 
	&\le[f(1)-f(1-\ep)]-[f(\ep)-f(0)]=\De_{f;d}(\ep).
\end{align*}
Thus, in the case $d=2$, \eqref{eq:sup} holds and hence \eqref{eq:} holds. 
This establishes the basis of the induction mentioned in the very beginning of the proof of Theorem~\ref{th:}. 

Assume now that $d\ge3$. 
By continuity and compactness, 
there is a maximizer $(p,q)\in P_{d;\ep}$ of $D_f$. 

\begin{framed}
\noindent In what follows, it is assumed by default that $(p,q)\in P_{d;\ep}$ is such a maximizer; \break 
in particular, we have $E(p,q)=\ep$. 
\end{framed}


Wlog, for some $k\in[d]$ we have 
\begin{gather}
	\label{eq:p ge q}
\text{$p_i\ge q_i$ for $i=1,\dots,k$,}	\\ 
\label{eq:p le q}
\text{$p_i\le q_i$ for $i=k+1,\dots,d$,}	\\ 
\label{eq:q decr}
q_1\ge\cdots\ge q_k. 
\end{gather}
So, 
\begin{equation*}
	\ep=\sum_{i=1}^k(p_i-q_i)=\sum_{i=k+1}^d(q_i-p_i). 
\end{equation*} 

Let now 
\begin{equation}\label{eq:p_i^*}
\text{$p_1^*:=q_1+\ep$ and $p_i^*:=q_i$ for $i=2,\dots,k$. 
}	
\end{equation}
Then the vector $(p_1^*,\dots,p_k^*)$ majorizes (in the Schur sense) the vector $(p_1,\dots,p_k)$. 

To see why this is true,  
note first that, by \eqref{eq:p_i^*}, \eqref{eq:q decr}, and \eqref{eq:p ge q}, 
$p^*_{k:1}=p^*_1=q_1+\ep$ and $p^*_{k:i}=q_{k:i}=q_i\le p_i$ for $i=2,\dots,k$. 
%
Moreover, by \eqref{eq:order} and \eqref{eq:p ge q}, 
$q_{k:i}\le p_{k:i}$ for all $i\in[k]$. So, 
\begin{equation}\label{eq:tails}
	\text{$\sum_{i=j+1}^k p^*_{k:i}=\sum_{i=j+1}^k q_{k:i}\le\sum_{i=j+1}^k p_{k:i}$ for all $j\in[k]$.}
\end{equation}
Also, 
\begin{equation}\label{eq:eq}
\sum_{i=1}^k p^*_{k:i}=\sum_{i=1}^k p^*_i=q_1+\ep+\sum_{i=2}^k q_i=\sum_{i=1}^k p_i. 	
\end{equation}
By \eqref{eq:eq} and \eqref{eq:tails},  
$\sum_{i=1}^k p^*_{k:i}=\sum_{i=1}^k p_i$ and $\sum_{i=1}^j p^*_{k:i}\ge\sum_{i=1}^j p_{k:i}$ for all $j\in[k]$. Thus, indeed $(p_1^*,\dots,p_k^*)\succcurlyeq(p_1,\dots,p_k)$. 

Therefore and because $f$ is continuous and convex, we have $\sum_1^k f(p^*_i)\ge\sum_1^k f(p_i)$, in view of the equivalence of items (i) and (iv) in \cite[Theorem~A.3]{marsh-ol09}. 
Also, by \eqref{eq:p_i^*}, $p_i^*\ge q_i$ for all $i\in[k]$. So, if we replace $p_1,\dots,p_k$ respectively by $p^*_1,\dots,p^*_k$, then condition \eqref{eq:p ge q} will continue to hold, as well as the other conditions imposed above on $p=(p_1,\dots,p_d)$, whereas the value of 
$D_f$, as defined in \eqref{eq:de}, may only increase after this replacement. 
 
So, wlog $(p_1,\dots,p_k)=(p_1^*,\dots,p_k^*)$. 
Then, by  
\eqref{eq:p_i^*} and \eqref{eq:p le q}, 
%
\begin{equation}\label{eq:break after 1}
	\text{$p_1
	>q_1$ and $p_i\le q_i$ for all $i=2,\dots,d$, }
\end{equation}
and hence
\begin{equation}\label{eq:ep=}
	\ep=p_1-q_1=\sum_2^d(q_i-p_i). 
\end{equation}


Take any permutation $\pi\colon[d]\to[d]$ such that $\pi(1)=1$. Let $p_\pi:=(p_{\pi(1)},\dots,p_{\pi(d)})$ and similarly define $q_\pi$. Then clearly $E(p_\pi,q_\pi)=E(p,q)$, $D_f(p_\pi,q_\pi)=D_f(p,q)$, and condition \eqref{eq:break after 1} holds with $(p_\pi,q_\pi)$ in place of $(p,q)$ whenever it holds for $(p,q)$. Let us refer to this as the permutation invariance (of $E$, $D_f$, and \eqref{eq:break after 1}).

Suppose now for a moment that $p_i=q_i$ for some $i\in[d]$; by 
the permutation invariance, 
wlog $i=d$, so that $p_d=q_d=:c
$. 
Let $\rho$ and $\si$ be two positive-semidefinite Hermitian linear operators acting on a Hilbert space of the finite dimension $d-1\ge2$, with the same eigenbasis and with the eigenvalues $p_1,\dots,p_{d-1}$ for $\rho$ and $q_1,\dots,q_{d-1}$ for $\si$. Then $\tr\rho=\tr\si=1-c$ and $T(\rho,\si)=p_1-q_1=\sum_2^{d-1}(q_i-p_i)=\sum_2^d(q_i-p_i)=\ep$, with $\ep$ as in \eqref{eq:ep=}. So, by Corollary~\ref{cor:t} (which is, specifically, a corollary of inequality \eqref{eq:}) and the induction mentioned in the beginning of the proof of Theorem~\ref{th:}, 
\begin{multline*}
|D_f(p,q)|=\Big|\sum_1^d f(p_i)-\sum_1^d f(q_i)\Big|
=\Big|\sum_1^{d-1} f(p_i)-\sum_1^{d-1} f(q_i)\Big| \\ 
=|S_f(\rho)-S_f(\si)|
\le\De_{1-c;f;d-1}(\ep)\le\De_{1;f;d}(\ep)=\De_{f;d}(\ep). 
\end{multline*} 
So, \eqref{eq:sup} follows if $p_i=q_i$ for some $i\in[d]$. 

Thus, we may and will henceforth assume that all inequalities in \eqref{eq:break after 1} are strict: 
\begin{equation}\label{eq:break after 1,str}
	\text{$p_1>q_1$ and $p_i<q_i$ for all $i\in[d]\setminus\{1\}$. }
\end{equation}

Suppose next that there are two distinct numbers $j$ and $k$ in $[d]\setminus\{1\}$ such that $p_j>0$ and $p_k>0$. By the permutation invariance, 
wlog 
$p_2\ge p_3>0$.
Replace now $p=(p_1,\dots,p_d)$ by $\tp:=(p_1,p_2+t,p_3-t,p_4,\dots,p_d)$, where 
$t>0$ is close enough to $0$ 
-- more specifically, one may take here any $t\in(0,\min[p_3,q_2-p_2])$. 
Then $(\tp,q)\in P_{d}$ and, by the condition $p_i<q_i$ for $i\in[d]\setminus\{1\}$ in \eqref{eq:break after 1,str},  
$E(\tp,q)=E(p,q)=\ep$, so that $(\tp,q)\in P_{d;\ep}$. 
Also, by the strict convexity of $f$, we have $D_f(\tp,q)>D_f(p,q)$, 
which contradicts the assumption that $(p,q)\in P_{d;\ep}$ is a maximizer of $D_f$. 

So, $p_{i_*}>0$ for at most one $i_*\in[d]\setminus\{1\}$, and, by the permutation invariance, 
wlog $i_*=2$, so that 
\begin{equation}\label{eq:p=}
	p=(p_1,p_2,0,\dots,0). 
\end{equation} 

Further, using the convexity of the function $f$ and Jensen's inequality, 
we see that $\sum_{i=3}^d f(q_i)\ge(d-2)f(\tfrac1{d-2}\,\sum_{i=3}^d q_i)$. So, wlog $q_3=\cdots=q_d$. 

Furthermore, if $q_2\ne q_3$, replace $q$ by $\tq:=(q_1,(1-t)q_2+tq_3,(1-t)q_3+tq_2,\break 
q_4,\dots,q_d)$ for a small enough $t\in(0,1)$. Then, in view of the condition $p_i<q_i$ for $i\in[d]\setminus\{1\}$ in \eqref{eq:break after 1,str},  $(p,\tq)\in P_{d;\ep}$, but $D_f(p,\tq)>D_f(p,q)$, 
a contradiction. 

Thus, $q_2=q_3=\cdots=q_d$, so that, in view of \eqref{eq:p=} and \eqref{eq:ep=}, for some $x\in[0,1]$ 
we have 
\begin{equation}\label{eq:PQ}
	p=P(x):=(x,1-x,0,\dots,0),\quad q=Q(x):=\Big(x-\ep,\frac{1-x+\ep}{d-1},\dots,\frac{1-x+\ep}{d-1}\Big).
\end{equation}
Moreover, 
condition $p_2<q_2$ in \eqref{eq:break after 1,str} can now be rewritten as  
\begin{equation}\label{eq:conds}
	\frac{1-x+\ep}{d-1}>1-x 
\end{equation}
and inequality \eqref{eq:sup} can be rewritten as 
\begin{equation}
h(x):=f(x)-f(x-\ep)+f(1-x)-(d-1)f\Big(\frac{1-x+\ep}{d-1}\Big)\overset{\text{(?)}}\le h(1), 
\end{equation}
which follows because
\begin{equation}
h'(u):=[f'(u)-f'(u-\ep)]+\Big[f'\Big(\frac{1-u+\ep}{d-1}\Big)-f'(1-u)\Big]\ge0 
\end{equation}
for $u\in[x,1]$, 
in view of the convexity of $f$ and condition \eqref{eq:conds} (which implies $\frac{1-u+\ep}{d-1}>1-u$ for all $u\in[x,1]$). 


This 
completes the proof of the inequality in \eqref{eq:}. 

To complete the entire proof of Theorem~\ref{th:}, it remains to note that 
the bound $\De_{f;d}(\ep)$ on $|S_f(\rho)-S_f(\si)|$ is attained, for each $\ep\in[0,1]$, when the density matrices $\rho$ and $\si$ have a common eigenbasis with respective $d$-tuples of eigenvalues $p=(p_1,\dots,p_d)=P(1)$ and $q=(q_1,\dots,q_d)=Q(1)$ with $P$ and $Q$ as defined in 
\eqref
{eq:PQ}. 
\end{proof}

\begin{proof}[Proof of Corollary~\ref{cor:}]
This follows immediately from Theorem~\ref{th:} and the observation, made in the paragraph containing inequality \eqref{eq:incr-decr}, that $\De_{f;d}(\ep)$ is nondecreasing in $\ep\in[0,1-\frac1d]$ and nonincreasing in $\ep\in[1-\frac1d,1]$.  
\end{proof}

{\bf Acknowledgment:} The author is grateful to the referee for a thorough reading of the paper, discovering an error in the first version of the proof of Theorem~\ref{th:}, and suggesting a significant simplification in the second version of the proof. 

{\bf Acknowledgment:} After this paper was accepted for publication and posted on arXiv,  
the author was informed by N.\ Datta about the paper  
\cite{hanson-datta2019}, where results more general than Theorem~\ref{th:} and Corollary~\ref{cor:} were obtained, in a somewhat different form, as well as about the related papers \cite{hanson-datta2017a,datta2017b}. However, the proofs of Theorem~\ref{th:} and Corollary~\ref{cor:} in this note seem to be shorter and more self-contained.






\bibliographystyle{abbrv}


\bibliography{C:/Users/ipinelis/Documents/pCloudSync/mtu_pCloud_02-02-17/bib_files/citations04-02-21}

\end{document}